\documentclass[pra,twocolumn,10pt,aps]{revtex4-1}
\usepackage[utf8]{inputenc}  
\usepackage[T1]{fontenc}     
\usepackage[british]{babel}  
\usepackage[scaled=0.86]{berasans}  
\usepackage[colorlinks=true,allcolors=blue]{hyperref}  
\usepackage{graphicx} 
\usepackage[babel]{microtype}  
\usepackage{amsmath,amssymb,amsthm,bm,amsfonts,mathrsfs,bbm} 
\usepackage{setspace}
\usepackage{braket}
\usepackage[bottom]{footmisc}
\usepackage{blochsphere}
\usepackage{pgfplots}
\usetikzlibrary{calc,3d,shapes, pgfplots.external, intersections}

\usepackage[bitstream-charter]{mathdesign}
\usepackage{tikz}
\usepackage{mathtools}
\usetikzlibrary{arrows.meta}
\usepackage{csquotes}

\usepackage{xspace}  
\usepackage{pgfplots}
\usepackage{verbatim}

\newcommand\id{\leavevmode\hbox{\small1\kern-3.3pt\normalsize1}}

\newtheorem{theorem}{Theorem}

\newtheorem{corollary}[theorem]{Corollary}

\begin{document}

\title{Impossibility of  Cloning  of Quantum Coherence }

\author{Dhrumil Patel\textsuperscript{1}}
\email{dhrumil.patel@research.iiit.ac.in}
\affiliation{Center for Security, Theory and Algorithmic Research, International Institute of Information 
Technology-Hyderabad, Gachibowli, Telangana-500032, India.}

\author{Subhasree Patro\textsuperscript{1}}
\email{subhasree.patro@research.iiit.ac.in}
\affiliation{Center for Security, Theory and Algorithmic Research, International Institute of Information 
Technology-Hyderabad, Gachibowli, Telangana-500032, India.}

\author{Chiranjeevi Vanarasa}
\email{chiranjeevi.v@research.iiit.ac.in}
\affiliation{Center for Security, Theory and Algorithmic Research, International Institute of Information 
Technology-Hyderabad, Gachibowli, Telangana-500032, India.}

\author{Indranil Chakrabarty}
\email{indranil.chakrabarty@iiit.ac.in}
\affiliation{Center for Security, Theory and Algorithmic Research, International Institute of Information 
Technology-Hyderabad, Gachibowli, Telangana-500032, India. \\
Centre for Theoretical Physics, Jamia Milia Islamia University, Jamia Nagar, Okhla, New Delhi, Delhi 110025.}

\author{Arun Kumar Pati}
\email{akpati@hri.res.in}
\thanks{\\\textsuperscript{1}These two authors contributed equally to this manuscript}
\affiliation{Quantum Information and Computation group, Harish-Chandra  Research  Institute, HBNI, Allahabad, India}

\date{\today}
\begin{abstract}
It is well known that it is impossible to clone an arbitrary quantum state. However, this inability does not lead directly to no-cloning of quantum coherence. Here, we show that it is impossible to clone the coherence of an arbitrary quantum state which is a stronger statement than the 'no-cloning of quantum state'. In particular, with ancillary system as machine state, we show that it is impossible to clone the coherence of states whose coherence is greater than the coherence of the known states on which the transformations are defined. Also, we characterize the class of states for which coherence cloning will be possible for a given choice of machine. Furthermore, we find the maximum range of states whose coherence can be cloned perfectly. The impossibility proof also holds when we do not include machine states.
\end{abstract}

\maketitle

\section{\label{sec:Intro}Introduction}
The phenomenon of quantum superposition and entanglement lies at the heart of quantum mechanics  acting as  resources which we can harness to perform practical and important information theoretic tasks \cite{Nielsen}. Motivated by the increasing  importance of quantum entanglement \cite{Horo} in quantum information processing and communication schemes, a general
study of the theory of resources within the paradigm of quantum mechanics  and beyond is being formulated. We have several entanglement measures to quantify entanglement, however, until recently there was no standard way to quantify the coherence present in a quantum state. 
Quantum coherence can be viewed as a fundamental signature of non classicality in physical systems. Coherence can also be used as a resource for certain tasks like better cooling \cite{Bras, Lost} or work extraction processes in nano-scale thermodynamics,  in many quantum algorithms \cite{Anand, Li, Cas}, in quantifying wave-particle duality \cite{WaveParticleDuality1, WaveParticleDuality2, WaveParticleDuality3} and in  biological processes \cite{band, wilde}. The  resource theory of quantum coherence \cite{Strlet3, Asboth, Xi, Strlet4, Qi, Chin, Zhu, Zhu1, Korz} along with other resource theories of entanglement and thermodynamics \cite{Hill, Nap} has also been established.
Once we have the measure based on a given set of axioms to quantify the coherence \cite{Aberg,Baum, Wint, Stret1, Stret2, Chit, Theu, Mukho, Mitch} we can build the resource theory of coherence. This seeks to quantify and study the amount of linear superposition a quantum state possesses with respect to a given basis. Given a state $\rho$, with its matrix elements as $\rho_{ij}$, the amount of coherence present in the state in the basis $\{|i\rangle\}$
is given by the quantity $C_l(\rho) = \sum_{i\neq j} |\langle i| \rho | j \rangle|$ which is known as the $l_1$-norm of coherence. Note that coherence is a basis dependent quantity as the amount of coherence will be different  in  different  basis.   Since the $l_1$-norm  is  a  function
of the off-diagonal elements of the given density matrix
representation, clearly the value of coherence will be zero in  the eigen basis of the density matrix, where there are no off-diagonal elements. 


Quantum superposition and entanglement play a pivotal role in achieving information processing tasks that are otherwise not possible by any other classical resource. The same properties also forbid us to do certain tasks that are otherwise achievable classically. 
It started with the no-cloning theorem which states that there does not exist any quantum operation which can perfectly duplicate a pure state \cite{clone}. In particular, the no-cloning theorem states that if we have cloning machine which can copy two orthogonal quantum states then with the same cloning machine it is impossible to create an identical copy of an arbitrary quantum state. Pati and Braunstein later showed that we cannot delete either of the two quantum states perfectly \cite{AKP}. In addition to these two famous no-cloning and no-deletion theorem there are many other no-go theorems like no-flipping (impossibility to flip an arbitrary quantum state) no-self replication (cannot have a universal quantum constructor ) \cite{AKP1}. A two dimensional quantum system can always be represented as points on the Bloch Sphere parametrized by azimuthal angle $\theta$ and the phase angle $\phi$. It is interesting to note that there are no-go theorems like no-partial erasure \cite{AKP2}, no-splitting \cite{Zhou}  and no- partial swapping \cite{Indranil} which together tells us the indivisibility of the information content present in a quantum system. 

At this point it is interesting to ask the question {\em whether it is possible to clone coherence of arbitrary quantum states.} We know that cloning of arbitrary quantum state implies signaling. Therefore, the no-signaling implies the no-cloning but the no-cloning does not imply the no-signaling. Since, the cloning of quantum states implies cloning of coherence, it is obvious that no-cloning of coherence implies no-cloning of quantum states. Therefore, cloning of coherence of an arbitrary state is a more fundamental question. Thus, the no-cloning of quantum coherence is more powerful than no-cloning of an arbitrary quantum state. We have given few examples of state cloners in the Appendix \ref{sec:App_CloningExamples} to illustrate that cloning of quantum coherence is not the same as the cloning of quantum state. 

In quantum mechanics the wave function describes the physical system completely. Hence, cloning of quantum states would mean cloning of both wave and particle aspects of the entity. However, when we clone coherence, we try to clone only the wave aspect. In this paper, we show that indeed it is so and these two cloners are different. It is interesting to see that we cannot clone the coherence of arbitrary superposition of orthogonal states as long as the coherence of the state is more than the coherence of the orthogonal states in the given basis. This result holds when we define the cloning transformation with the machine states. However, we cannot say directly anything specific when the coherence of the input state is less than or equal to the coherence of these orthogonal states but nevertheless in this zone we are able to characterize the states  whose coherence can be cloned. We also find the maximum range of states whose coherence can be cloned perfectly.
Further, we show that there does not exist any universal unitary operator as a coherence cloner even when we are not considering the ancillary states. We find that the impossibility of universal coherence cloning fundamentally depends upon the choice of the known states and is very much  different from the cloning of the quantum state.   

\section{\label{sec:QCohCloning}No-Cloning of Quantum Coherence With Machine States}
In the case of the no-cloning theorem for quantum states we start with an assumption that we can clone two known  orthogonal quantum states. Here, we start with an assumption that we can copy the coherence of two known orthogonal quantum states $|\psi_{1}\rangle$ and $|\psi_{2}\rangle$ and then prove that it is impossible to clone the coherence of an unknown quantum state universally. At this point one may ask the question that how do we know that we can clone coherence of two orthogonal quantum states $|\psi_{1}\rangle$ and $|\psi_{2}\rangle$. We can argue that since $|\psi_{1}\rangle$ and $|\psi_{2}\rangle$ are two known orthogonal states, we can make copies of these quantum states. Now, cloning of quantum states always imply that the cloning of coherence is true (though the reverse is not true). This is because when we can clone the entire state we can definitely clone the coherence content of the state. Therefore, it is natural to assume that we can clone coherence of two known orthogonal states.

Let $U_{cc}$ be the unitary transformation that produces two copies of coherence starting from two orthogonal quantum states. The cloning transformation for coherence is given by,
\begin{equation}
    \label{eq:pure2mixed1}
    \begin{aligned}
        |\psi_{1}\rangle_A |0\rangle_B |X_{0}\rangle_C \longrightarrow |\Psi_{1}{}\rangle_{AB} |X_{1}\rangle_C,\\
        |\psi_{2}\rangle_A |0\rangle_B |X_{0}\rangle_C \longrightarrow |\Psi_{2}{}\rangle_{AB} |X_{2}\rangle_C,
    \end{aligned}
\end{equation}
where $|\psi_{1}\rangle$, $|\psi_{2}\rangle$ are input states, $|0\rangle$ is the blank state and $|X_{0}\rangle$ is the initial machine state. Also, $|\Psi_{1}\rangle$ and $|\Psi_{2}\rangle$ are states whose subsystems $A$ and $B$ have coherence same as that of the input states, and, $|X_{1}\rangle$ and $|X_{2}\rangle$ are the corresponding final machine states. The machine states satisfy $\langle X_{2}| X_{1}\rangle$ = 0 due to unitarity of the transformation.
Let us represent the two orthogonal states $|\psi_{1}\rangle_{A}$ and $|\psi_{2}\rangle_{A}$ in the $\{|0\rangle,|1\rangle\}$ basis as, $|\psi_{1}\rangle_{A} = a|0\rangle_{A} + b|1\rangle_{A}$ and $|\psi_{2}\rangle_{A} = b^{*}|0\rangle_{A} - a^{*}|1\rangle_{A}$.

As the transformation demands coherence to be perfectly copied, we must have $C_l(|\psi_{1}\rangle_{A}$) = $C_l(\rho_{A}{'})$ = $C_l(\rho_{B}{'})$ and $C_l(|\psi_{2}\rangle_{A})$ =  $C_l(\rho_{A}{''})$ = $C_l(\rho_{B}{''}) $ where,
\begin{equation}
    \begin{aligned}
        \rho_{A}{'} &= Tr_{B}|\Psi_{1}\rangle_{ABAB} \langle \Psi_{1}|,  \rho_{B}{'} = Tr_{A}|\Psi_{1}\rangle_{ABAB} \langle \Psi_{1}|, \\
        \rho_{A}{''} &= Tr_{B}|\Psi_{2}\rangle_{ABAB} \langle \Psi_{2}|,  \rho_{B}{''} = Tr_{A}|\Psi_{2}\rangle_{ABAB} \langle \Psi_{2}|.
    \end{aligned}
\end{equation}
Since the coherence of the orthogonal states are same, we have $C_l(|\psi_{1}\rangle_{A})=C_l(|\psi_{2}\rangle_{A})= 2|a||b|$, where $C_l(\rho)$ is the $l_1$-norm for quantifying quantum coherence. It may be noted that in the case of cloning of quantum states we require two identical copies of the input state at the output port. However, for cloning of coherence this is not the case as there can be two non identical state with the same coherence.
Since any state can be represented on the Bloch sphere as $ \rho = \frac{I + \vec{m}^{\,}.\vec{\sigma}^{\,}}{2}  $ with $ \vec{m}^{\,} = (m_{x}, m_{y}, m_{z})  $ as the Bloch vector and $\vec{\sigma}^{\,} = (\sigma_{x}, \sigma_{y}, \sigma_{z}) $ are the Pauli matrices. The coherence in $\{|0\rangle$, $|1\rangle\}$ basis is given by $C_{l}(\rho) = \sqrt{m_{x}^{2} + m_{y}^{2}}$. Hence, coherence only depends on $m_{x}$ and $m_{y}$ values. As shown in Fig. \ref{fig:TwoColors}, we can say that all the states that lie on the curved surface of the cylinder with radius  $\sqrt{m_{x}^{2} + m_{y}^{2}}$ will have the same coherence. 
\noindent At this point it is important to ask this question: \textit{Does quantum mechanics allow existence of a universal cloner for cloning the coherence of an arbitrary input state $\alpha|\psi_{1}\rangle + \beta|\psi_{2}\rangle$. } The answer to the question is \textit{No}.

\begin{theorem}
\label{th:noUniversalCloner}
It is impossible to clone the coherence of an arbitrary quantum state $|\psi\rangle$=$\alpha|\psi_{1}\rangle +\beta|\psi_{2}\rangle$, with the cloning transformations given by equation (\ref{eq:pure2mixed1}) when the coherence of the state $|\psi\rangle$ is more than the coherence of the states $|\psi_i\rangle$ $(i=1,2)$ for a fixed choice of basis $\{|0\rangle, |1\rangle\}$ .
\end{theorem}
\begin{proof}
Without the loss of generality, let us use the $l_{1}$-norm as a measure of quantum coherence and assume fixed basis as the computational basis. Any arbitrary state in $|\psi_{1}\rangle_{A}$, $|\psi_{2}\rangle_{A}$ basis can be written as $|\psi\rangle_{A} =   \alpha|\psi_{1}\rangle_{A} + \beta|\psi_{2}\rangle_{A}$. The
$l1$-norm of coherence of the state $|\psi\rangle_{A}$ in $\{|0\rangle, |1\rangle\}$ basis is $2|(\alpha a + \beta b^{*})(\alpha b - \beta a^{*})|$. After the application of cloning transformation $U_{cc}$,the arbitrary state along with the blank and machine states becomes ($\alpha|\Psi_{1}\rangle_{AB}|X_{1}\rangle_{C} + \beta|\Psi_{2}\rangle_{AB}|X_{2}\rangle_{C}$). Tracing out the subsystems B and C, we get $\rho_{A}^{final}=|\alpha|^{2}\rho_{A}{'} + |\beta|^{2}|\rho_{A}{''}$. From the convexity property of coherence measure we have, $ C_l(\rho_{A}^{final}) \leq$ $(|\alpha|^{2}C_l(\rho_{A}{'}) + |\beta|^{2}C_l(\rho_{A}{''}))$ = $2(|\alpha|^{2}|a||b| + |\beta|^{2}|a||b|)$ = $2|a||b|=C(|\psi_i\rangle)$. Therefore, the final coherence of the subsystem A is at most $2|a||b|$. Therefore, it is evident that all the input states $|\psi\rangle$ whose initial coherence $C_l(|\psi\rangle)$ is greater than $2|a||b|$, which is the coherence of the known orthogonal states, it is impossible to clone the coherence perfectly. 
\end{proof}

\noindent \textbf{Note 1:} Theorem 1 holds for all coherence measures and is not only restricted to the $l_1$-norm of coherence. The convexity of any coherence measure ensures that the final coherence $C(\rho_{A}^{final})$ is bounded above by $ C(|\psi_{i}\rangle_{A})$, where $i=1,2$.\\

\noindent 
This tells that if the coherence of an arbitrary input state is greater than the coherence of the orthogonal states then we cannot copy the coherence of the state into a blank state. Geometrically, if we consider the Bloch sphere as the state space, the orthogonal states   $|\psi_1\rangle$ and $|\psi_2\rangle$ represent two symmetric points on the surface of each hemisphere of the Bloch sphere. Taking the shortest distance of each of these points from the central axis as radius, these circles will represent all the states with same coherence value. We will have exactly two similar circles, one in each hemispheres representing orthogonal states. All the pure states with greater coherence value will be the points on the surface which are lying between these two circles.
This theorem geometrically tells us that we cannot copy the coherence of the intermediate surface points (see Fig. \ref{fig:TwoColors}). However, the theorem does not tell anything about the points on the surface which lies on the circles (except $|\psi_1\rangle$ and $|\psi_2\rangle$ ) and other points lying between those circles and poles. It may be possible to clone some of these states. The theorem only tells us that given a choice of known orthogonal states there does not exist any universal cloner which will clone all pure states on the surface of the Bloch sphere. However, the theorem is only true as long as the orthogonal states are not from the equatorial circle of the Bloch sphere ($|+\rangle$ and $|-\rangle$ lying on the equator of the Bloch sphere can be one such example). In that scenario, we do not have any input state with a coherence greater than the coherence of these equatorial orthogonal states (which is 1). If we view this, circles from each hemisphere coincides with each other and there is no intermediate point. The important question is whether for such choice of cloner it is possible to clone all the states on the surface of Bloch sphere. The answer to this is once again "no" and indeed there does not exist a universal cloner for whatever choice of machine.

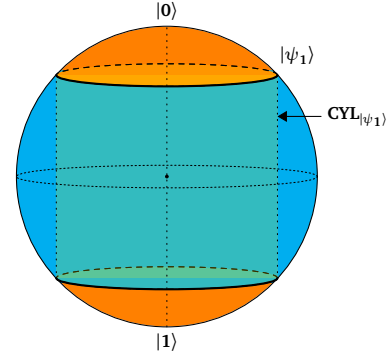
\begin{figure}
\begin{tikzpicture}[line cap=round, line join=round, >=Triangle]
  \tikzset{
    partial ellipse/.style args={#1:#2:#3}{
        insert path={+ (#1:#3) arc (#1:#2:#3)}
    }
  }
  \clip(-3,-3) rectangle (3.2,3);
  \scriptsize
  
  \draw [draw = white, fill = orange] (0,0) circle (2cm);
  \begin{scope}[even odd rule]
  \clip (1.414,1.414)  -- (-1.414,1.414) arc(135:225:2cm) -- (1.414,-1.414) arc(315:360:2cm) -- (2,0) arc(0:45:2cm);
  \fill [fill = cyan] (-3,-3) rectangle (3,3);
  \end{scope}
  
  \draw [dotted, fill = cyan] [rotate around={0:(0,0)}, dash pattern=on 2pt off 2pt] (0,-1.35) ellipse (1.47cm and 0.15cm);
  \draw[dotted, fill = orange] [rotate around={0:(0,0)}, dash pattern=on 2pt off 2pt] (0,1.35) ellipse (1.47cm and 0.15cm);
  
  \draw [dotted, fill = yellow, fill opacity = 0.2] [rotate around={0:(0,0)}, dash pattern=on 2pt off 2pt] (0,-1.35) ellipse (1.47cm and 0.15cm);
  \draw [draw = white, fill = yellow, opacity = 0, fill opacity =0.2] (1.47,1.35) rectangle (-1.47,-1.35);
  \draw[dotted, fill = yellow,  fill opacity = 0.2] [rotate around={0:(0,0)}, dash pattern=on 2pt off 2pt] (0,1.35) ellipse (1.47cm and 0.15cm);
  
  \draw[thick,black] (0,-1.35) [partial ellipse=180:360:1.47cm and 0.15cm];
  \draw[thick,black] (0,1.35) [partial ellipse=180:360:1.47cm and 0.15cm];
  \draw (0,0) circle (2cm);
  \draw (1.414,1.414) node[anchor=south west] {$\mathbf{ |\psi_1\rangle}$};

  \draw [rotate around={0.:(0.,0.)},dash pattern=on 0.5pt off 1pt] (0,0) ellipse (2cm and 0.15cm);
  \draw [dotted] (0,-2) node[anchor=north] {$\mathbf {|1\rangle}$} -- (0,2) node[anchor=south] {$\mathbf {|0\rangle}$};
  \draw [dotted] (1.47,1.35) -- (1.47,-1.35);
  \draw [dotted] (-1.47,1.35) -- (-1.47,-1.35);
  \draw [->] (2.05,0.8) node[anchor=west] {$\mathbf {CYL_{|\psi_1\rangle}}$} -- (1.47, 0.8) ;
  
  \draw [fill] (0,0) circle (0.5pt);

\end{tikzpicture}
\caption{Classification of states for a given cloner. The blue zone on the surface of the Bloch sphere denotes the pure states whose coherence can not be cloned. The remaining zone is showed by orange colour. $CYL_{|\psi_1\rangle}$ represents all the states (pure as well as mixed) which have same coherence as $|\psi_1\rangle$}
\label{fig:TwoColors}
\end{figure}



\begin{corollary}
For a cloning transformation given by equation (\ref{eq:pure2mixed1}), with the choice of known orthogonal states $|\psi_{1(E)}\rangle$ and $|\psi_{2(E)}\rangle$, taken from the equator, it is impossible to clone the coherence of an arbitrary quantum state $|\psi\rangle$=$\alpha|\psi_{1(E)}\rangle +\beta|\psi_{2(E)}\rangle$, for a fixed choice of basis $\{|0\rangle, |1\rangle\}$ .
\end{corollary}
\begin{proof}
As per the assumption that the unitary cloning transformation ($U_{cc}$) perfectly clones the coherence of $|\psi_{1(E)}\rangle_{A}$, $|\psi_{2(E)}\rangle_{A}$ into subsystem of $|\Psi_{1}\rangle_{AB}$, $|\Psi_{2}\rangle_{AB}$ we have $2|x|$ = 1 and $2|y|$ = 1, where, $ x = \rho_{A 01}^{'}$ and $y = \rho_{A 01}^{''}$ are the off-diagonal terms of the subsystem $\rho_{A}{'}$ and $\rho_{A}{''}$ respectively. Here we will only look at system A and prove that there exists some states for which cloning is not possible. With the constraint $2|a||b|=1$ and the normalization condition we have $|a|$ = $|b|$ = $\frac{1}{\sqrt{2}}$. Now, the initial coherence of the input state becomes $C_l(|\psi\rangle_{A})=2|\alpha^2 ab - \beta^2 a^{*}b^{*}|=2\sqrt{\frac{1}{4}-2Re(\alpha^2\beta^{*2}(ab)^2)}$. However, the final coherence of the subsystem $A$ is given by,
$C_l(\rho^{final}_{A})=2|(|\alpha|^{2}x + |\beta|^{2}y)|$.

Let us assume that there is a universal machine that clones coherence of any arbitrary state $|\psi\rangle$. Then $C_l(|\psi\rangle)$ should be equal to $C_l(\rho^{final}_{A})$ for every $\alpha$ and $\beta$ values. But we can see that there exist $\alpha$ and $\beta$ values such that the initial and final coherence values are not equal.

There exist $\alpha_{1}$ and $\beta_{1}$ such that $|\alpha|=|\alpha_{1}|$ and $|\beta|=|\beta_{1}|$ but $\alpha\neq\alpha_{1}$ or/and $\beta\neq\beta_{1}$. In that case though the final coherences will be equal, the initial coherences are not. That clearly means perfect cloning of coherence is not possible for at least one of those two states.
\end{proof}

\noindent\textbf{Example :}
Let, $|\chi_1\rangle = \alpha_1|\psi_1\rangle + \beta_1|\psi_2\rangle$ and $|\chi_2\rangle = \alpha_2|\psi_1\rangle + \beta_2|\psi_2\rangle$
where, $\alpha_1=\frac{1}{\sqrt{2}}$, $\beta_1=\frac{1}{\sqrt{2}}$, $\alpha_{2}=\frac{\imath}{\sqrt{2}}$, $\beta_{2}=\frac{1}{\sqrt{2}}$. Here $|\psi_1\rangle$ and $|\psi_2\rangle$ are the states defined in the Eq. \ref{eq:pure2mixed1}. Then initial coherence of $|\chi_1\rangle$ is given by $C_l(|\chi_1\rangle) = |ab - a^{*}b^{*}|$ and that of $|\chi_2\rangle$ is $C_l(|\chi_2\rangle) = |ab + a^{*}b^{*}|$ given that |a| = |b| = $\frac{1}{\sqrt{2}}$. Final coherence, $C_l(\rho_1^{final})$ = $C_l(\rho_2^{final})$ = |$\rho_{A01}^{'}$ + $\rho_{A01}^{''}$|, where $\rho_{A01}^{'}$ and $\rho_{A01}^{''}$ are the off diagonal terms of the susbsystem $\rho_A^{'}$ and $\rho_A^{''}$, respectively.
\newline
Clearly we see that there is a mismatch of the initial coherence of states $|\chi_1\rangle$ and $|\chi_2\rangle$, but their final coherence is same. Therefore, at least for one of the states the coherence is not getting perfectly copied. 

\section{\label{sec:SCloning}No Cloning of Quantum Coherence Without Machine States}

In the previous section, we have shown that there does not exist universal cloning transformation which will be able to clone the coherence of any arbitrary state. In the previous proof the cloning transformation includes the ancilla states representing the machine states. In  this section, we investigate whether there exists any unitary in general which will act on the input state and blank state without invoking ancillary state that will clone coherence for any arbitrary state. We find that there exists no such unitary. Like in the previous section, here also we assume that the perfect cloning is possible for two known orthogonal states $|\psi_{1}\rangle$ and $|\psi_{2}\rangle$. The transformation is given by

\begin{equation}
    \label{eq:pure2mixedStronger}
    \begin{aligned}
        |\psi_{1}\rangle_A |0\rangle_B  \longrightarrow |\Psi_{1}\rangle_{AB}, \\
        |\psi_{2}\rangle_A |0\rangle_B \longrightarrow |\Psi_{2}\rangle_{AB},
    \end{aligned}
\end{equation}
where, $\langle\psi_{1}|\psi_{2}\rangle$ = 0. Therefore, $\langle\Psi_{1}|\Psi_{2}\rangle$ = 0.
\begin{theorem}
\label{th:noStrongUniversalCloner}
It is impossible to clone the coherence of any arbitrary quantum state $|\psi\rangle$=$\alpha|\psi_{1}\rangle +\beta|\psi_{2}\rangle$, with the cloning transformations given by Eq. \ref{eq:pure2mixedStronger}.
\end{theorem}
\begin{proof}
Let us assume that there exists a unitary that clones coherence of any arbitrary quantum state. Then this unitary should clone coherence for the states $|+\rangle$ and $|-\rangle$ as well. As these states are maximally coherent states, and this machine can clone the coherence perfectly, then the output states should also be maximally coherent states. The transformation would be given by Eq. \ref{eq:plusMinusTransformation}. The output states in this case are all pure because there are no mixed states whose coherence can be 1.
\begin{equation}
    \label{eq:plusMinusTransformation}
    \begin{aligned}
        |+\rangle_A |0\rangle_B  \longrightarrow |\psi_{1}{'}\rangle_{A}|\psi_{1}{''}\rangle_{B} \\
        |-\rangle_A |0\rangle_B \longrightarrow |\psi_{2}{'}\rangle_{A}|\psi_{2}{''}\rangle_{B}
    \end{aligned}
\end{equation}
 where, either $\langle\psi_{1}{'}|\psi_{2}{'}\rangle$ = 0 or $\langle\psi_{1}{''}|\psi_{2}{''}\rangle$ = 0. Let $|\phi\rangle = \gamma|+\rangle + \delta|-\rangle$ be an arbitrary quantum state on the equatorial circle of the Bloch sphere. The transformation as given in Eq. \ref{eq:plusMinusTransformation} results in the state of the system $AB$ to $|\Phi_{final}\rangle = \gamma|\psi_{1}{'}\rangle_{A}|\psi_{1}{''}\rangle_{B} + \delta|\psi_{2}{'}\rangle_{A}|\psi_{2}{''}\rangle_{B}$. For the coherence to be cloned perfectly, $|\Phi_{final}\rangle$ needs to be a separable system of two maximally coherent states. This makes either $|\psi_{1}{'}\rangle_{A} = |\psi_{2}{'}\rangle_{A}$ or $|\psi_{1}{''}\rangle_{A} = |\psi_{2}{''}\rangle_{A}$ because one of this pair has to be orthogonal.
\newline
Without loss of generality lets assume that $\langle\psi_{1}{'}|\psi_{2}{'}\rangle$ = 0 and $|\psi_{1}{''}\rangle_{A} = |\psi_{2}{''}\rangle_{A}$ then any state $|\psi_1\rangle$ with $C_l(|\psi_1\rangle) < 1$ from Eq. \ref{eq:pure2mixedStronger} can be written as $\alpha|+\rangle + \beta|-\rangle$, but under this transformation rules the system will transform to $(\alpha|\psi_{1}{'}\rangle + \beta|\psi_{2}{'})_A|\psi_{1}{''}\rangle_B$. Though the coherence of subsystem A is preserved the coherence of subsystem B is still 1. 
\end{proof}

\textit{                       }

\section{\label{sec:CONCL}CONCLUSION}

To summarize, we have shown that no-cloning for quantum states does not lead to the no-cloning of coherence. In fact, we prove much stronger statement, that is, the no-cloning of quantum coherence implies the no-cloning of quantum state. Since coherence captures the wave aspect of quantum particles our result shows that quantum wave cannot be amplified, where as classical wave can be amplified. This brings out a fundamental difference between classical and quantum wave. In particular, we are able to show that, for the input state whose coherence is greater than the coherence of the known states, coherence cloning is not possible. Besides this, we characterize the class of states for which coherence cloning
will be possible for a given choice of machine built on
known orthogonal states and find the maximum range of states whose coherence can be cloned perfectly. Interestingly, we also show that the universal cloner does not exist even in the situation where we have no ancillary inputs.\\ 

\textit{Acknowledgement:} Authors acknowledge  Dr. S. Mitra and Dr. K. Srinathan for having various fruitful discussions at different stage in the process of completing this research article.

\begin{appendix}
\section{\label{sec:App_CloningExamples}How cloning of coherence is different from cloning of states}
Here, we illustrate some instances where coherence cloning is not the same as the cloning of quantum states. As an example, consider the Wootter-Zurich (WZ) cloning machine \cite{clone} that performs the following operation
\begin{equation}
\label{eq:WootterZurich}
\begin{split}
|0\rangle_A |0\rangle_B |X_0\rangle_C \longrightarrow |0\rangle_A |0\rangle_B |X_1\rangle_C,\\
|1\rangle_A |0\rangle_B |X_0\rangle_C \longrightarrow |1\rangle_A |1\rangle_B |X_2\rangle_C.
\end{split}
\end{equation}
Here A, B and C are the input, output and machine qubits respectively. When we apply the same W-Z cloning machine on an arbitrary quantum state $\alpha|0\rangle + \beta|1\rangle$ ($|\alpha|^{2}+ |\beta|^{2} = 1 $) whose coherence is $2|\alpha||\beta|$ in $\{|0\rangle,|1\rangle\}$ basis, it transforms to the state $\alpha^2|0\rangle\langle0|+\beta^2|1\rangle\langle 1|$ which has zero coherence in $\{|0\rangle,|1\rangle\}$ basis. This shows that even if the state gets cloned approximately with W-Z cloning machine, there is no cloning of the coherence as W-Z cloning machine does not take into account the off-diagonal terms of the state. 

Let us consider another example of the Buzek Hillery (BH) cloning machine \cite{BuzekHilleryCloner} which is proved to be optimal state independent quantum cloning machine \cite{GisinOptimalCloningFidelity}. The 2-dimensional BH cloning transformation is given as
\begin{equation}
\begin{split}
  \left|\Psi_1\right\rangle_{A} \left|0\right\rangle_{B} \left|X_0\right\rangle_C &\rightarrow  c\left|\Psi_1\right\rangle_{A} \left|\Psi_1\right\rangle_{B} \left|X_{11}\right\rangle_C \nonumber\\
&+ d \left(\left|\Psi_1\right\rangle_{A} \left|\Psi_2\right\rangle_{B} +\left|\Psi_2\right\rangle_{A}\left|\Psi_1\right\rangle_{B}\right) \left|Y_{12}\right\rangle_C,
\\
\left|\Psi_2\right\rangle_{A} \left|0\right\rangle_{B} \left|X_0\right\rangle_C &\rightarrow  c\left|\Psi_2\right\rangle_{A} \left|\Psi_2\right\rangle_{B} \left|X_{22}\right\rangle_C \nonumber\\
&+ d \left(\left|\Psi_2\right\rangle_{A} \left|\Psi_1\right\rangle_{B} +\left|\Psi_1\right\rangle_{A}\left|\Psi_2\right\rangle_{B}\right) \left|Y_{21}\right\rangle_C,
\label{eq:B-H_gen_transform}  
\end{split}
\end{equation}
where the coefficients $c$ and $d$ are real. The notation A, B and C represents the input, output and machine qubits respectively. In case of cloning a single qubit, using the no-signaling constraint and the fidelity as parameter of quantum cloning machine, Gisin proved that B-H state independent quantum cloner is the optimal one with the fidelity $\frac{5}{6}$ \cite{GisinOptimalCloningFidelity}. But if we consider the ratio of the final coherence to the initial coherence ($l_1-$ norm), the B-H cloner gives $\frac{2}{3}$. That is, in other words, two thirds of coherence is getting copied with the B-H cloner. This example shows that even though information gets cloned upto $\frac{5}{6}$, coherence gets cloned only upto $\frac{2}{3}$. Possibly, this suggests us that when we clone quantum information we try to clone both the wave information and the particle information. As BH machine only clones wave information upto $\frac{2}{3}$ it may be the case that the higher value of $\frac{5}{6}$ is due to particle nature getting cloned more compared to wave information.

\begin{figure}
\begin{tikzpicture}[line cap=round, line join=round, >=Triangle]
 \clip(-3,-3) rectangle (3.2,3);
  
 \draw(0,0) circle (2cm);
  
  \draw [rotate around={0.:(0.,0.)},dash pattern=on 0.5pt off 1pt] (0,0) ellipse (2cm and 0.15cm);
  \tikzset{
    partial ellipse/.style args={#1:#2:#3}{
        insert path={+ (#1:#3) arc (#1:#2:#3)}
    }
  }
  \scriptsize
  
  \draw[blue] [rotate around={0:(0,0)}] (0,-1.35) ellipse (1.47cm and 0.15cm);
  \draw[blue] (1.462,1.364) node[anchor=south west] { $\mathbf{|\psi_1\rangle}$} -- (1.462,-1.364);
  \draw[blue] (-1.462,1.364) -- (-1.462,-1.364);
  \draw[blue] [rotate around={0:(0,0)}] (0,1.35) ellipse (1.47cm and 0.15cm);
  
  \draw[brown] [rotate around={0:(0,0)}] (0,-1.75) ellipse (0.95cm and 0.1cm);
  \draw[brown] (0.969,1.749) node[anchor=south west] { $\mathbf{|\psi\rangle}$} -- (0.969,-1.749);
  \draw[brown] (-0.969,1.749) -- (-0.969,-1.749);
  \draw[brown] [rotate around={0:(0,0)}] (0,1.75) ellipse (0.95cm and 0.1cm);
  
  
  \draw [dotted] (0,-2) -- (0,2);
  \draw (0,2) node[anchor=south] {$\mathbf {|0\rangle}$};
  \draw (0,-2) node[anchor=north] {$\mathbf {|1\rangle}$};
  \draw (-1.462,0.8) node[anchor=east] {$\mathbf {\rho^{'}_A}$} -- (1.462,-0.4) node[anchor=west] {$\mathbf {\rho^{''}_A}$};
  \draw (-1.462,-0.8) node[anchor=south east] {$\mathbf {\rho^{'}_B}$} -- (1.462,0.4) node[anchor=west] {$\mathbf {\rho^{''}_B}$};
  \draw [->] (2.05,0.8) node[anchor=west] {$\mathbf{CYL_{|\psi_1\rangle}}$} -- (1.462, 0.8) ;
  \draw [->] (2.05,-0.8) node[anchor=west] {$\mathbf{CYL_{|\psi\rangle}}$} -- (0.969, -0.8) ;
  \draw [fill] (-0.969,0.6) node[anchor=south west] {$\mathbf {\rho^{f1}_A}$} circle (0.75pt);
  \draw [fill] (-0.969,-0.6) node[anchor=north west] {$\mathbf {\rho^{f1}_B}$} circle (0.75pt);
  \draw [fill] (0.969,0.2) node[anchor=south east] {$\mathbf {\rho^{f2}_B}$} circle (0.75pt);
  \draw [fill] (0.969,-0.2) node[anchor=north east] {$\mathbf {\rho^{f2}_A}$} circle (0.75pt);
  \draw [fill] (0,0) circle (0.5pt);
  
\end{tikzpicture}
\caption{Solutions with convex combinations: The blue cylinder represents the states having the same coherence as the known orthogonal states. The orange cylinder represents the states having the same coherence as the input states. Here $\rho_{A(B)}^{f1(f2)}$ are nothing but $\rho_{A(B)}^{final1(final2)}$.}
\label{fig:TwoCylinders}
\end{figure}
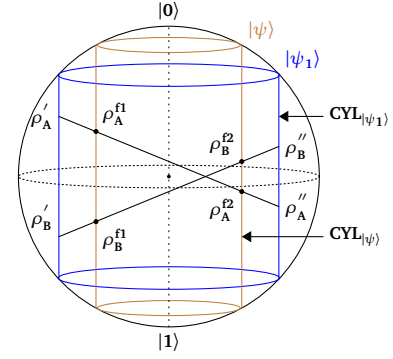

\section{\label{sec:CHARAC}Classification of states given a coherence cloner}

In this subsection, we try to characterize the states whose coherence can be perfectly cloned given a machine defined over $|\psi_1\rangle$ and $|\psi_2\rangle$. Geometrically, we attempt to find out points on the surface of the sphere for which the cloning of coherence is possible. The entire Bloch sphere can be divided in two zones namely $C_{l}(|\psi\rangle_{A}) \leq 2|a||b|$  and $C_{l}(|\psi\rangle_{A}) > 2|a||b|$. In theorem \ref{th:noUniversalCloner}, we have already shown that cloning of coherence is not possible when $C_{l}(|\psi\rangle_{A}) > 2|a||b|$, however, it is not clear when $C_{l}(|\psi\rangle_{A}) \leq 2|a||b|$.

Let us take an arbitrary state $|\psi\rangle$ from the orange zone as shown in the Fig. \ref{fig:TwoColors}. All the states both pure and mixed that have same coherence value as $|\psi\rangle$ lie on $CYL_{|\psi\rangle}$ as shown in the Fig. \ref{fig:TwoCylinders}. For the coherence of $|\psi\rangle$ to be perfectly cloned, the output states $\rho_{A}^{final}$ and $\rho_{B}^{final}$ should lie on $CYL_{|\psi\rangle}$. As we have seen earlier $\rho_{A}^{final}=|\alpha|^{2}\rho_{A}{'} + |\beta|^{2}|\rho_{A}{''}$ and $\rho_{B}^{final}=|\alpha|^{2}\rho_{B}{'} + |\beta|^{2}|\rho_{B}{''}$ are convex combination of $\rho_{A}{'}$, $\rho_{A}{''}$ and $\rho_{B}{'}$, $\rho_{B}{''}$ respectively. Here, $\rho_{A}{'}$, $\rho_{A}{''}$ and $\rho_{B}{'}$, $\rho_{B}{''}$ are the mixed output states for the known orthogonal states $|\psi_i\rangle$ and should lie on the wider cylinder $CYL_{|\psi_1\rangle}$. This would mean that $\rho_{A}^{final}$ is an intersection of the line segment joining $\rho_{A}{'}$, $\rho_{A}{''}$ and $CYL_{|\psi\rangle}$, similar condition must hold for $\rho_{B}^{final}$.

For perfect cloning to happen the line segment joining $\rho_A^{'}$ and $\rho_A^{''}$ and the line segment joining $\rho_B^{'}$ and $\rho_B^{''}$ should intersect $CYL_{|\psi\rangle}$ in equal proportions, as it is evident from the expressions of $\rho_{A}^{final}$ and $\rho_{B}^{final}$. Let us imagine that $\rho_{A(B)}^{final1(final2)}$ are four intersection points. Fig. \ref{fig:TopView} shows some of the possible orientations of these four points.

Without loss of generality, let us just look at the subsystem A. Let $|\alpha| = k$. To find all the pure states that have same coherence as $|\psi\rangle$, whose coherence can be perfectly cloned, we only need to see what are the points of intersection of rims of $CYL_{|\psi\rangle}$ and the circle $CIRC_k$, where $CIRC_k$ contains all the points $\alpha|\psi_1\rangle + \beta|\psi_2\rangle$ whose $|\alpha| = k$, as shown in Fig. \ref{fig:IntersectingCYLcirc}. Depending on the $CYL_{|\psi\rangle}$, $CYL_{|\psi_1\rangle}$ and the values of k, the number of points of intersections will vary from 0 to 4, as shown in the Fig. \ref{fig:IntersectingCYLcirc}. Similarly, when $|\alpha| = 1-k$ we get the same number of solutions. Therefore, the total number of solutions vary as 0, 2, 4, 6 or 8.

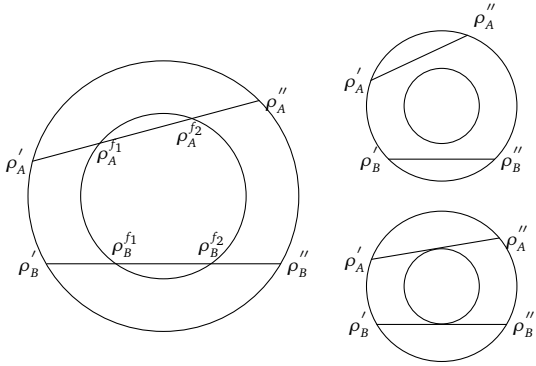
\begin{figure}
\begin{tikzpicture}[line cap=round, line join=round, >=Triangle]
  \clip(-4,-3.5) rectangle (3.5,3.5);
  \scriptsize
  
  \draw (-1.7,0) circle (1.8cm);
  \draw (-1.7,0) circle (1.1cm);
  \draw (-3.438,0.465) node[anchor=east] {$\rho^{'}_A$} -- (-0.427,1.272) node[anchor=west] {$\rho^{''}_A$};
  \draw (-3.25,-0.9) node[anchor=east] {$\rho^{'}_B$} -- (-0.14,-0.9) node[anchor=west] {$\rho^{''}_B$};
  \draw (-2.4,0.85) node[anchor=north] {$\rho^{f_1}_A$};
  \draw (-1.34,1.08) node[anchor=north] {$\rho^{f_2}_A$};
  \draw (-2.2,-0.9) node[anchor=south] {$\rho^{f_1}_B$};
  \draw (-1.05,-0.9) node[anchor=south] {$\rho^{f_2}_B$};

  \draw (2,1.2) circle (1cm);
  \draw (2,1.2) circle (0.5cm);
  \draw (1.06,1.54) node[anchor=east] {$\rho^{'}_A$} -- (2.342,2.14) node[anchor=south west] {$\rho^{''}_A$};
  \draw (1.292,0.492) node[anchor=east] {$\rho^{'}_B$} -- (2.70,0.492) node[anchor=west] {$\rho^{''}_B$};

  \draw (2,-1.2) circle (1cm);
  \draw (2,-1.2) circle (0.5cm);
  \draw (1.066,-0.841) node[anchor=east] { $\rho^{'}_A$} -- (2.766,-0.557) node[anchor=west] {$\rho^{''}_A$};
  \draw (1.138,-1.707) node[anchor=east] { $\rho^{'}_B$} -- (2.861,-1.707) node[anchor=west] {$\rho^{''}_B$};
  
\end{tikzpicture}
\caption{Top view to depict existence of solutions}
\label{fig:TopView}
\end{figure}

\begin{figure}
\begin{tikzpicture}[line cap=round, line join=round, >=Triangle]
  \clip(-2.70,-2.70) rectangle (2.66,2.58);
  \scriptsize
  
  \draw(0,0) circle (2cm);
  
  \draw [orange,thick][rotate around={115:(0,0)},dash pattern=on 2pt off 2pt] (0,1.78) ellipse (0.86cm and 0.1cm);
  \draw [orange,thick][rotate around={295:(0,0)},dash pattern=on 2pt off 2pt] (0,1.78) ellipse (0.86cm and 0.1cm);
  \filldraw [fill=black, draw=black] (-1.259,-1.554) circle (0.7pt);
  \filldraw [fill=black, draw=black] (1.259,1.554) circle (0.7pt);
  \draw [yellow,thick][rotate around={115:(0,0)},dash pattern=on 2pt off 2pt] (0,1) ellipse (1.72cm and 0.1cm);
  \draw [yellow,thick][rotate around={295:(0,0)},dash pattern=on 2pt off 2pt] (0,1) ellipse (1.72cm and 0.1cm);
  \filldraw [fill=black, draw=black] (0.38,1.7) circle (0.7pt);
  \filldraw [fill=black, draw=black] (-0.38,-1.7) circle (0.7pt);
  \filldraw [fill=black, draw=black] (0.365,1.375) circle (0.7pt);
  \filldraw [fill=black, draw=black] (-0.365,-1.375) circle (0.7pt);
  \draw [cyan,thick][rotate around={115:(0,0)},dash pattern=on 2pt off 2pt] (0,0.50) ellipse (1.92cm and 0.1cm);
  \draw [cyan,thick][rotate around={295:(0,0)},dash pattern=on 2pt off 2pt] (0,0.50) ellipse (1.92cm and 0.1cm);
  \filldraw [fill=black, draw=black] (-0.18,1.7) circle (0.7pt);
  \filldraw [fill=black, draw=black] (0.18,-1.7) circle (0.7pt);
  \filldraw [fill=black, draw=black] (1.272,-1.543) circle (0.7pt);
  \filldraw [fill=black, draw=black] (-1.272,1.543) circle (0.7pt);
  \filldraw [fill=black, draw=black] (-0.165,1.375) circle (0.7pt);
  \filldraw [fill=black, draw=black] (0.165,-1.375) circle (0.7pt);
  \draw [blue,thick][rotate around={115:(0,0)},dash pattern=on 2pt off 2pt] (0,0.15) ellipse (1.985cm and 0.1cm);
  \draw [blue,thick][rotate around={295:(0,0)},dash pattern=on 2pt off 2pt] (0,0.15) ellipse (1.985cm and 0.1cm);
  \filldraw [fill=black, draw=black] (-0.57,1.69) circle (0.7pt);
  \filldraw [fill=black, draw=black] (0.57,-1.69) circle (0.7pt);
  \filldraw [fill=black, draw=black] (-0.56,1.39) circle (0.7pt);
  \filldraw [fill=black, draw=black] (0.56,-1.39) circle (0.7pt);
  
  \filldraw [fill=black, draw=black] (-0.9,1.655) circle (0.7pt);
  \filldraw [fill=black, draw=black] (0.9,-1.655) circle (0.7pt);
  \filldraw [fill=black, draw=black] (-0.885,1.415) circle (0.7pt);
  \filldraw [fill=black, draw=black] (0.885,-1.415) circle (0.7pt);
  
  \draw [rotate around={0:(0,0)},dash pattern=on 2pt off 2pt] (0,1.54) ellipse (1.25cm and 0.17cm);
  \draw [rotate around={0:(0,0)},dash pattern=on 2pt off 2pt] (0,-1.54) ellipse (1.25cm and 0.17cm);
  
  \draw (1.813,0.845)-- (-1.813,-0.845);
  \draw (-0.3,2.6) node[anchor=north west] {$\mathbf {|0\rangle}$};
  \draw (-0.3,-2) node[anchor=north west] {$\mathbf {|1\rangle}$};
  \draw (1.813,0.845) node[anchor=south west] {$|\psi_1\rangle$};
  \draw (-1.813,-0.845) node[anchor=north east] {$|\psi_2\rangle$};
  
  \draw [fill] (0,0) circle (0.5pt);
  
\end{tikzpicture}
\caption{Different cases of possible solutions: The figure shows depending on the choice of k different $CIRC_k$s (the colored ones) intersect the rims of the $CYL_{|\psi\rangle}$ (the two black ones) , the number of points of intersections will vary from 0 to 4. Total number of possible solutions will vary as 0, 2, 4, 6 or 8 }
\label{fig:IntersectingCYLcirc}
\end{figure}
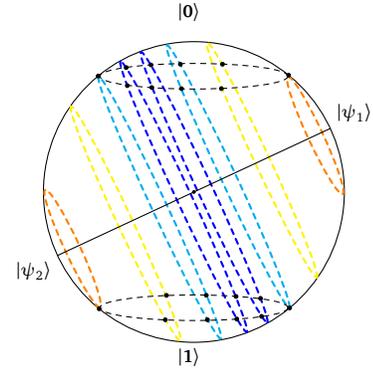

\section{\label{sec:Maximization}Maximization of coherence cloners}
In the earlier section, we have seen classification of states given a particular coherence cloner. It is clear that the range of states whose coherence can be cloned perfectly depends on the cloner $U_{cc}$ . In this section, we discuss the techniques to maximize this cloner so as to have more range of states whose coherence can be cloned perfectly.

There can be infinite number of unitaries that can be defined based on the transformation rules defined in Eq. \ref{eq:pure2mixed1}. Every unitary depends on six states $|\psi_1\rangle_{A}$, $|\psi_2\rangle_{A}$, $\rho_A^{'}$, $\rho_A^{''}$, $\rho_B^{'}$, $\rho_B^{''}$. 

The cloner $U_{cc}$ on system ABC transforms subsystem A which was an arbitrary quantum state   $|\psi\rangle=\alpha|\psi_1\rangle + \beta|\psi_2\rangle$ and subsystem B which was a blank state $|0\rangle$ to $\rho_A^{final} = |\alpha|^2\rho_A^{'} + |\beta|^2\rho_A^{''}$ and  $\rho_B^{final} = |\alpha|^2\rho_B^{'} + |\beta|^2\rho_B^{''}$ respectively. $\rho_{A(B)}^{final}$ is a convex combination of $\rho_{A(B)}^{'}$ and $\rho_{A(B)}^{''}$. Therefore, final state $\rho_{A(B)}^{final}$ should lie somewhere on the line segment joining the states $\rho_{A(B)}^{'}$ and $\rho_{A(B)}^{''}$ in the Bloch sphere, as we can see in the Fig. \ref{fig:TwoCylinders}. Therefore, we can say that possibility of perfect coherence cloning for a state depends on whether the line segments intersect the cylinder with $C_l({\rho_{A(B)}^{final}})$ or not. Given the fact that they do intersect, they have to intersect with the same ratio as each other, only then perfect cloning will be possible on both the subsystems otherwise we can definitely say that the $C_l(|\psi\rangle)$ cannot be perfectly copied.

This brings us to our first level of maximization of our cloner $U_{cc}$. We can see in the Fig. \ref{fig:TopView} that the cloners which have their line segments joining $\rho_A^{'}$ and $\rho_A^{''}$ and line segment joining $\rho_B^{'}$ and $\rho_B^{''}$ pass through the central axis allow for the possibility of perfect coherence cloning for a bigger range of states as they will intersect all the cylinders above them.

The second level of maximization can be done in the following way. We can see that if the starting states $|\psi_{1}\rangle_A$ and $|\psi_{2}\rangle_A$ of the assumed cloner lies on the equatorial plane, i.e. $C_l(|\psi_{1}\rangle$) = $C_l(|\psi_{2}\rangle)$ = $1$  and the output states will have equal coherence to that of the starting states, then this cloner will give maximum number of perfectly cloned copies as it will intersect all the cylinders on the sphere. Then the cloning transformation is given by
\begin{equation}
\label{eq:MaxCloner1}
\begin{split}
|\psi_{1}\rangle_A |0\rangle_B|X_{0}\rangle_C \longrightarrow |\psi_{1}{'}\rangle_A |\psi_{1}{''}\rangle_B|X_1\rangle_C,\\
|\psi_{2}\rangle_A |0\rangle_B|X_{0}\rangle_C \longrightarrow |\psi_{2}{'}\rangle_A |\psi_{2}{''}\rangle_B|X_2\rangle_C.
\end{split}
\end{equation}
Here $\langle\psi_{1}^{'}|\psi_2^{'}\rangle=0$ and $\langle\psi_{1}^{''}|\psi_2^{''}\rangle=0$ as this would ensure that the line segment joining $\rho_A^{'}$ and $\rho_A^{''}$ and the line segment joining $\rho_B^{'}$ and $\rho_B^{''}$ both pass through the central axis. 

Interestingly, it is observed that class of states whose coherence can be cloned perfectly given the transformations defined in Eq. \ref{eq:MaxCloner1} and the conditions $\langle\psi_{1}^{'}|\psi_2^{'}\rangle=0$ and $\langle\psi_{1}^{''}|\psi_2^{''}\rangle=0$ are the states that lie on the great circle passing through the states $|\psi_1\rangle$, $|\psi_2\rangle$, $|0\rangle$ and $|1\rangle$ on the Bloch sphere. 

The calculations are as follows: The states $|\psi_1\rangle=\frac{1}{\sqrt{2}}|0\rangle + \frac{1}{\sqrt{2}}\mathrm{e}^{\imath\phi_1}|1\rangle$ and $|\psi_2\rangle=\frac{1}{\sqrt{2}}|0\rangle + \frac{1}{\sqrt{2}}\mathrm{e}^{\imath(\phi_1 + \pi)}|1\rangle$ represent a pair of orthogonal states on the equatorial circle of the Bloch sphere. Then, any arbitrary state $|\psi\rangle = \alpha|\psi_1\rangle + \beta|\psi_2\rangle$ can be written as $\cos{\frac{\theta}{2}}|0\rangle + \sin{\frac{\theta}{2}}\mathrm{e}^{\imath\phi_2}|1\rangle$ in $\{|0\rangle, |1\rangle\}$ basis. Then, $C_l(|\psi\rangle)=|\sin{\theta}|$ in $\{|0\rangle, |1\rangle\}$ basis.
As $\alpha=\frac{1}{\sqrt{2}}(\cos{\frac{\theta}{2}}+\sin{\frac{\theta}{2}}\mathrm{e}^{\imath(\phi_2 - \phi_1)})$ and $\beta=\frac{1}{\sqrt{2}}(\cos{\frac{\theta}{2}}-\sin{\frac{\theta}{2}}\mathrm{e}^{\imath(\phi_2 - \phi_1)})$, the final coherence which is given by $C_l(\psi^{final})=||\alpha|^2-|\beta|^2|=|2|\alpha|^2-1|$ becomes $|\sin{\theta}\cos{(\phi_2-\phi_1)}|$. 

We see that the only solution where initial coherence is equal to final coherence is when $\phi_1=\phi_2$. Therefore, for all values of $\theta$ the final coherence $C_l(\psi^{final})=C_l(\psi)$ if $\phi_1=\phi_2$. Which means that the cloner defined in equations \ref{eq:MaxCloner1} perfectly clones coherence for all the states on the great circle passing through $|\psi_1\rangle$, $|\psi_2\rangle$, $|0\rangle$, $|1\rangle$ as shown in the Fig. \ref{fig:OptimalSolution}.


\begin{figure}
\begin{tikzpicture}[line cap=round, line join=round, >=Triangle]
 \clip(-3,-3) rectangle (3.2,3);
  
 \draw(0,0) circle (2cm);
 \scriptsize
  \draw [rotate around={0.:(0.,0.)},dash pattern=on 1pt off 1pt] (0,0) ellipse (2cm and 0.15cm);
  \tikzset{
    partial ellipse/.style args={#1:#2:#3}{
        insert path={+ (#1:#3) arc (#1:#2:#3)}
    }
  }
  \draw[blue,thick] [rotate around={0:(0,0)}] (0,0) ellipse (1cm and 2cm);
  
  
  \draw [dotted] (0,-2) -- (0,2);
  \draw (0,2) node[anchor=south] {$\mathbf {|0\rangle}$};
  \draw (0,-2) node[anchor=north] {$\mathbf {|1\rangle}$};
  \draw [fill = blue] (0.99,0.14) node[anchor=south west] {$|\psi_1\rangle$} circle (1.25pt);
  \draw [fill = blue] (-0.99,-0.14) node[anchor=north east] {$|\psi_2\rangle$} circle (1.25pt);
  \draw [fill] (0,0) circle (0.5pt);
\end{tikzpicture}
\caption{States whose coherence can be perfectly cloned}
\label{fig:OptimalSolution}
\end{figure}
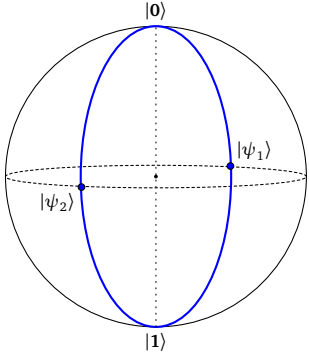

\end{appendix}
\end{document}